 \documentclass[a4paper,12pt]{amsart}
 \usepackage{amsmath,amsthm,amssymb}
 \usepackage{graphicx}
 \usepackage[all]{xy}
 \usepackage{txfonts}
\usepackage{hyperref}
\usepackage[margin=2cm]{geometry}

\theoremstyle{definition}
\newtheorem{thm}{Theorem}
\newtheorem{lem}{Lemma}
\newtheorem{prop}{Proposition}

\newtheorem{cor}{Corollary}
\newtheorem{conj}{Conjecture}
\newtheorem{dfn}{Definition}
\newtheorem{rem}{Remark}
\newtheorem{ex}{Example}

\newcommand{\B}{\mathbb{B}}
\newcommand{\Sl}{\mathbb{S}}

\newcommand{\Net}{\mathrm{Evo}}

\newcommand{\csentence}{\textbf}

\newcommand{\doiurl}[1]{\href{http://dx.doi.org/#1}{\nolinkurl{#1}}}

\begin{document}
\title{A circuit-preserving mapping from multilevel to Boolean dynamics}

\author{Adrien Faur\'e}
\address[A. Faur\'e]{Department of Physics and Information Science \endgraf
Yamaguchi University  \endgraf 
1677-1, Yoshida, Yamaguchi 753-8512, Japan}
\email{afaure@yamaguchi-u.ac.jp}
\author{Shizuo Kaji}
\address[S. Kaji]{Department of Mathematics \endgraf
Yamaguchi University  \endgraf 
1677-1, Yoshida, Yamaguchi 753-8512, Japan \endgraf
/ JST PRESTO}
\thanks{The authors contributed equally to this work.}
\email{skaji@yamaguchi-u.ac.jp}

\date{}

\maketitle

\begin{abstract}
Many discrete models of biological networks rely exclusively on Boolean variables
and many tools and theorems are available for analysis of {strictly} Boolean models.
However, multilevel variables are often required to account for threshold effects,
in which knowledge of the Boolean case does not generalise straightforwardly.
This motivated the development of conversion methods for multilevel to Boolean models.
In particular, Van Ham's method has been shown to yield a one-to-one, neighbour and regulation preserving dynamics, making it the de facto standard approach to the problem.
However, Van Ham's method has several drawbacks: most notably, it introduces vast regions of ``non-admissible'' states
 that have no counterpart in the multilevel, original model. 
 This raises special difficulties for the analysis of interaction between variables and circuit functionality, which is believed to be central to the understanding of dynamic properties of logical models.
Here, we propose a new multilevel to Boolean conversion method, with software implementation.
Contrary to Van Ham's, our method doesn't yield a one-to-one transposition of multilevel trajectories;
however, it maps each and every Boolean state to a specific multilevel state, thus getting rid of the non-admissible regions and, at the expense of (apparently) more complicated, ``parallel'' trajectories.
One of the prominent features of our method is that it preserves dynamics and interaction of variables in a certain manner.
As a demonstration of the usability of our method,
we apply it to construct a new Boolean counter-example to the well-known conjecture that a local negative circuit is necessary to generate sustained oscillations.
This result illustrates the general relevance of our method for the study of multilevel logical models.

\end{abstract}

\section{Background}
%Discrete 
{Boolean} models % are proved to be
{have proved} very useful in the analysis of various networks in biology.
{However, it is often convenient to introduce multilevel variables to account for multiple threshold effects.}
We are often faced with choices between using Boolean variables or multilevel variables.
This can be crucial since theoretical results are sometimes proved only for Boolean or multilevel networks.
A particular example of this situation is in {Ren\'e Thomas' conjecture} %the theory initiated by Thomas on the relation between
%the existence of sustained oscillations and that of feedback loops in the gene interaction in the gene regulatory network.
 that a local negative circuit is necessary to produce sustained (asynchronous) oscillations.
This paper stems from the simple idea that 
a Boolean counter-example to {that} conjecture
could be found by transposing a multilevel counter-example found earlier by Richard and Comet.
However, we believe the method developed in this paper, together with a handy script which implements it,
 is widely applicable to other theoretical studies which involves 
discrete networks.
We also find the notion of asymptotic evolution function defined in this paper sheds light on the understanding of
relation between the state transition graph and the interaction graph.

%%%%%%%%%%%%%%%%%%%%%%%%
\subsection{Introduction}
%%%%%%%%%%%%%%%%%%%%%%%%

Introduced in the 1960s-70s to model {\em biological regulatory networks}, the logical (discrete) formalism has gained increasing popularity, with recent applications as diverse as drosophila development, cell cycle control, or immunology (see \cite{Abou-Jaoude2016} for a survey).
%\blue{\cite{Jong2002} too old?}).
While many of these models rely exclusively on Boolean variables, 
{it is often} useful to introduce multilevel variables % at times 
to account for more refined behaviour.
However,
many tools and theoretical results are restricted to the Boolean case
(see e.g. \cite{Stoll2012, MacNamara2012, Helikar2012})
This situation motivated the development of methods to convert multilevel models to Boolean ones \cite{Remy2006, VanHam}.
A simple idea for such a conversion was introduced by Van Ham \cite{VanHam}, and this method has been shown to be essentially 
the only one that could provide a ``one-to-one, neighbour and regulation preserving map'' \cite{Didier2011}.
One problem with the conversion is that the resulting Boolean model is 
defined only on a %the 
sub-region of the whole Boolean state space, %which is 
called the \emph{admissible region},
and how to extend the model outside {that} %the admissible 
region is not trivial.
This leads to potential problems with analytical tools designed to deal with the whole state space, as a property that is true in the restricted domain may be false on the whole state space, and vice versa.
The primary goal of the present paper is to address this issue by introducing an extension of Van~Ham's method.
More precisely, we introduce a new method for multilevel to Boolean model conversion
which extends the domain of Van Ham's model to the whole state space
while preserving edge functionality and, therefore, local circuits.
Our mapping yields a state transition graph with ``parallel'' trajectories
that contains the one obtained by Van Ham's mapping as a sub-graph in such a way that attractors of the dynamics are preserved. 

We apply our method to investigate a particular class of theoretical results that connect the asynchronous behaviour of a model to the presence of
 \emph{regulatory circuits} in the interaction graph. 
 In the early 1980s, R. Thomas conjectured that the presence of a positive circuit (\emph{i.e.} a circuit where each component directly or indirectly has a positive effect upon itself) in the interaction graph is a necessary condition for multi-stability, and a negative circuit (where each component has a negative effect upon itself) is necessary for sustained oscillations \cite{Thomas1981}.
 One particular formulation of the conjecture focuses on local or ``type-1'' circuits \cite{Comet2013}, i.e. circuits whose arcs are all functional in the same single point of the system's state transition graph -- as opposed to global circuits whose arcs may be functional anywhere. While the conjecture holds for positive circuits both at the global and local levels, and for multilevel as well as Boolean models \cite{Remy2008, Richard2007}, in the negative case the conjecture could only be proved true at the global level \cite{Remy2008}. At the local level, a counter-example has first been published for multilevel models \cite{Richard2010}, while the Boolean case remained open \cite{Comet2013} until a Boolean counter-example was eventually discovered \cite{Ruet2017}, showing that contrary to expectations, a local negative circuit was not necessary to generate sustained oscillations. Interestingly, the approaches taken by P. Ruet and A. Richard are rather different, and their counter-examples have little in common.
Applying our method to the Richard-Comet multilevel counter-example,
we obtain a new Boolean counter-example to the conjecture that a local negative circuit is necessary to produce sustained oscillations.
 
%%%%%%%%%%%%%%%%%%%%%%%%
\subsection{Definitions}
%%%%%%%%%%%%%%%%%%%%%%%%
\subsubsection{Evolution function and State transition graph}
We work within the generalised logical framework introduced by Ren\'e Thomas and collaborators~\cite{Thomas1990}; see Abou-Jaoud\'e \emph{et al.}~\cite{Abou-Jaoude2016} for a recent review.
Here, we introduce the notation we use throughout this paper. 
Fix positive integers $n$ and $m_i \ (1\le i\le n)$.
Consider a system consisting of mutually interacting $n$ genes, indexed by the set $I=\{1,2,\ldots,n\}$.
Each gene $a_i$ takes expression levels in the integer interval $\{0,1,\ldots,m_i\}$.
The state of the system evolves depending on the current state.
This leads to a discrete dynamical system represented by a {\em evolution function over $M$}
\[
f=(f_1,f_2,\ldots,f_{n}): M \to M,
\]
 where $M=\left\{ (x_1,\ldots,x_n)\mid x_i\in \{0,1,\ldots,m_i\} \right\}$.
As a special case when $m_i=1$ for all $i\in I$, we denote $M=\B^n$ with $\B=\{0,1\}$ and call the system {\em Boolean}.
A basic question asks what we can tell about the asymptotic global behaviour of the dynamics,
which is encoded in the \emph{state transition graph},
 from local data of $f$, which are encoded in the \emph{partial derivatives} of $f$ or the \emph{interaction graph}.

%In this paper, we consider directed graphs, 
%and edges, paths, and cycles are always assumed to be directed.

The evolution of the whole system can be formally modelled by a certain kind of directed graph on $M$.
We equip $M$ with the usual metric
$d(x,x')=\sum_{i=1}^{n} |x_i-x'_i|$ for $x,x'\in M$.
Denote by $e_1=(1,0,0,\ldots), e_2=(0,1,0,0,\ldots),$ etc. the coordinate vectors of $M$. 
A {\em grid graph $\Gamma$ over $M$} is a graph with the vertex set $M$ satisfying that
\begin{itemize}
\item each directed edge connects a pair of vertices of distance one
\item at each vertex $x$ there are no two parallel outward edges; that is, 
$x-e_j \leftarrow x \to x+e_j$
 is not allowed.
\end{itemize}
The state of the whole system is represented by the levels of genes, %so
{and} corresponds to a vertex in $\Gamma$.
At each time step, the state evolves to one of its neighbouring vertices connected by an arrow in the following way.
%The second condition prevents contradicting evolution.
To an evolution function over $M$, we associate a grid graph $\Gamma(f)$ over $M$
called the {\em (asynchronous) state transition graph}
with the edge set 
\begin{equation}\label{eq:stg}
\left\{(x_1,x_2,\ldots,x_j,\ldots,x_{n}) \to (x_1,x_2,\ldots,x_j+\delta,\ldots,x_{n}),
\delta=\begin{cases} -1 & (f_j(x)<x_j) \\ +1 & (f_j(x)>x_j) \end{cases} \right\}.
\end{equation}
Note that here we follow the standard convention that transition of states is \emph{unitary} (see \cite[\S4]{Richard2010})
so that the existence of an edge $x\to x'$ implies $d(x,x')=1$;
that is, at each step the level of a single gene changes at most by one.

Asymptotic behaviour of the evolution of a system can be captured in a graph theoretical entity of the state transition graph.
An {\em attractor} is a terminal strongly connected sub-graph of $\Gamma$;
that is, any two elements of it are connected by a path and there is no edge from its elements to one in the complement.
An attractor consisting of a single vertex is called a {\em stable state}, 
otherwise it is called a {\em cyclic attractor}.
Intuitively, attractors are domains in $\Gamma$ in which the system eventually resides;
there is no way to escape once the system arrives in it, but each state in the domain can be visited after arbitrarily many steps.

%%%%%%%%%%%%%%%%%%%%
\subsubsection{Interaction graph and circuit functionality}
%%%%%%%%%%%%%%%%%%%%
A common practice in analysing interactions among genes in a network
is to encode it in the form of a labelled directed graph called the interaction graph,
where interaction is measured by the partial derivatives of the evolution function $f=(f_1,f_2,\ldots,f_n): M\to M$.

The {\em forward partial derivative} of $f_i$ along the $j$-th coordinate at $x=(x_1,\ldots,x_n)$ with $x_j<m_j$ is defined by
\[
\partial^+_j f_{i}(x)=f_i(x_1,\ldots,x_j+1,\ldots,x_{n})-f_i(x_1,\ldots,x_j,\ldots,x_{n})
=f_i(x+e_j)-f_i(x).
\]
The {\em backward partial derivative} along the $j$-th coordinate at $x$ with $x_j>0$ is defined similarly by
\[
\partial^-_j f_{i}(x)=f_j(x_1,\ldots,x_j,\ldots,x_{n})-f_j(x_1,\ldots,x_j-1,\ldots,x_{n})
=f_i(x)-f_i(x-e_j).
\]
Partial derivatives $\partial^+_j f_{i}(x)$ and $\partial^-_j f_{i}(x)$
are non-trivial when the $i$-th gene's target value changes along the change of the $j$-th gene.
They encode the dependence between genes locally at the state $x\in M$.

\begin{rem}
For a Boolean network, only one of the forward or the backward partial derivative exists at each $x$, so we just put them together to define the ordinary partial derivative denoted by $\partial_j$. 
On the other hand, in multilevel case, we have both the forward and the backward partial derivatives at some $x$.
It is important to consider both of them (c.f. \cite[Definition 8]{Richard2010}).
\end{rem}

%\subsection*{Local interaction graph and circuit functionality}
\begin{dfn}
The (local) {\em interaction graph} $Gf(x)$ of $f$ at $x$ is a graph over
the vertex set $I$ such that there exists an edge from $j$ to $i$ 
\begin{itemize}
\item with label ``$+$'' if $\partial^+_j f_{i}(x)>0$ or $\partial^-_j f_{i}(x)>0$
\item with label ``$-$`` if $\partial^+_j f_{i}(x)<0$ or $\partial^-_j f_{i}(x)<0$.
\end{itemize}
Note that we can have both positive and negative edges from $j$ to $i$ at the same time.
We define the {\em global interaction graph} $Gf(M)$ as the union of edges of $Gf(x)$ for all $x\in M$.
\end{dfn}

\begin{dfn}[{\cite[Definition 1, Proposition 1]{Comet2013}}]{\ }\\
\begin{itemize}
\item
A cycle $C$ in $Gf(M)$ is called a positive (resp. negative) circuit 
if it contains an even (resp. odd) number of negative edges.
\item
A circuit $C$ is said to be {\em type-1 functional} if $C\subset Gf(x)$ for some $x$.
%\item A cycle $C$ with vertices $J\subset I$ is said to be {\em type-2 functional}
%if $\exists z\in M^{I\setminus J}$ such that $C\subset Gf(y,z)$ for all $y\in M^J$.
\end{itemize}
\end{dfn}
%The existence of circuit implies closed chain of interaction or feedback loop in the system.

As in the continuous case, a function $f$ 
is recovered up to constant by its partial derivatives:
For two evolution functions $f,g: M\to M$ satisfying
$\partial^+_j f_i=\partial^+_j g_i$ for all $i,j\in I$
(or $\partial^-_j f_i=\partial^-_j g_i$ for all $i,j\in I$), 
the difference $f_i(x)-g_i(x)$ is constant for any $i\in I$.
In particular, two distinct Boolean evolution functions have the same partial derivatives if and only if
they are constant and do not coincide at any point.
This means the partial derivatives have almost all the information of the network %.
However, the next example shows that 
%\red{and it is trivial to show that} 
the partial derivatives 
are not enough to determine the asymptotic behaviour of the dynamics.
\begin{ex}\label{ex:J}
Consider the the Boolean evolution functions defined by
%described by the following mod-$2$ polynomials 
%(see \cite{Veliz-Cuba} for the polynomial notation)
\[
f(x_1,x_2) = \begin{cases} (0,0) & ((x_1,x_2)=(1,0))\\ (1,0) & (\text{otherwise}) \end{cases} \quad
g(x_1,x_2) =  \begin{cases} (0,1) & ((x_1,x_2)=(1,0))\\ (1,1) & (\text{otherwise}). \end{cases}
\]
Since they differ by a constant, their partial derivatives agree.
There exists the unique cyclic attractor $(0,0) \leftrightarrow (1,0)$ in $\Gamma(f)$,
whereas there exists the unique stable state $(1,1)$ in $\Gamma(g)$.
\end{ex}

%%%%%%%%%%%%%%%%%
\section{Methods}
%%%%%%%%%%%%%%%%%
\subsection{Asymptotic evolution function}
The correspondence between evolution functions and state transition graphs is not bijective.
In fact, as discussed by Streck \emph{et al.}~\cite{Streck2015}, for a given (multilevel) grid graph $\Gamma$, there are multiple evolution functions which have $\Gamma$ as 
their state transition graphs.
To have a bijective correspondence between the two representations of the system,
we restrict ourselves to a certain class of evolution functions.
There are two major conventions:
\begin{enumerate}
\item We say $f$ is {\em stepwise} or {\em unitary} if
$|f_i(x_1,\ldots,x_{n})-x_i| \le 1$
for all $i\in I$ and $x\in M$.
\item We say $f$ is {\em asymptotic} if
$f_i(x_1,\ldots,x_i,\ldots,x_{n})\in \{0,x_i,m_i\}$ for all $i\in I$ and $x\in M$.
\end{enumerate}
In both cases, $f_i$ encodes only the sign of $f_i(x)-x_i$.
%%%

For any evolution function, there exists a unique asymptotic and a unique stepwise evolution functions
having the same state transition graph.
\begin{prop}\label{graph-function-correspondence}
For any evolution function $f$, define
\[
\bar{f}_i(x)=\begin{cases} m_i & (f_i(x)>x_i) \\ x_i & (f_i(x)=x_i) \\ 0 & (f_i(x)<x_i) \end{cases} \quad
\hat{f}_i(x)=\begin{cases} x_i+1 & (f_i(x)>x_i) \\ x_i & (f_i(x)=x_i) \\ x_i-1 & (f_i(x)<x_i). \end{cases}
\]
Then, $\bar{f}$ is asymptotic and $\hat{f}$ is stepwise with $\Gamma(f)=\Gamma(\bar{f})=\Gamma(\hat{f})$.

Similarly, for any grid graph $\Gamma$, there exists a unique asymptotic  
function $\bar{f}^\Gamma$ and a unique stepwise function $\hat{f}^\Gamma$ such that 
$\Gamma=\Gamma(\bar{f}^\Gamma)=\Gamma(\hat{f}^\Gamma)$.
\end{prop}
\begin{proof}
We see how to define $\bar{f}^\Gamma$.
For a vertex $x\in \Gamma$ and $i\in I$, 
we have only one of the three possibilities:
\begin{itemize}
\item $x\to x-e_i$
\item $x\to x+e_i$
\item there is no edge from $x$ in the direction of $e_i$.
\end{itemize}
We define an asymptotic evolution function $\bar{f}^\Gamma$ by setting $f^\Gamma_i(x)=0, m, x_i$ accordingly.
\end{proof}

If we are interested in the evolution of a system, which is encoded in the state transition graph, we can
restrict ourselves to either the class of stepwise evolution functions or the class of asymptotic evolution functions.
Our choice in this paper is to restrict ourselves to the latter, and 
we identify an asymptotic evolution function with its state transition graph and vice versa.
In the rest of the paper, we assume functions are asymptotic unless otherwise stated and denoted just by $f$
without a bar over it.
 
There is a little difference in the interaction graph when we consider the stepwise case instead of the asymptotic case.
When $i\neq j$, $\partial^+_i f_j(x)$ and $\partial^+_i \hat{f}_j(x)$ have the same sign,
and same is true for the backward partial derivatives.
However, when $i=j$, we have the following difference.
\begin{ex}
Consider the asymptotic evolution function $f_1(0)=2, f_1(1)=2, f_1(2)=2$ over $M=\{0,1,2\}$.
The corresponding stepwise evolution function is $\hat{f}_1(0)=1, \hat{f}_1(1)=2, \hat{f}_1(2)=2$.
At $x=1$, there is no edges in the interaction graph of $f$ while
there is a positive self-loop in the one of $\hat{f}$.
On the other hand, consider the asymptotic evolution function $f_1(0) = 2, f_1(1) = 1, f_1(2) = 0$.
At $x=1$, there is a negative self-loop in the interaction graph of $f$,
 whereas there is no arrow in the one of the corresponding stepwise evolution function
 $\hat{f}_1(0)= 1, \hat{f}_1(1) = 1, \hat{f}_1(2) = 1$.
\end{ex}
In short, the interaction graphs of an asymptotic function and its corresponding stepwise function are the same only up to self-loops.

A function which is neither asymptotic nor stepwise has
in general more non-trivial partial derivatives than the asymptotic and the stepwise functions
sharing the same state transition graph given in Proposition \ref{graph-function-correspondence}.
(See \cite{Streck2015} for a detailed discussion. The stepwise function in our paper 
can be seen as a special case of the canonical function defined there.)
\begin{ex}
An asymptotic function $f: \{0,1,2\}^2 \to \{0,1,2\}^2$ defined by
$(f_1,f_2)(x_1,x_2)=(2,2)$ have the same state transition graph with
\[
(g_1,g_2)(x_1,x_2)=\begin{cases} (1,2) & (x_1,x_2)=(0,0) \\ (2,2) & (\text{otherwise}) \end{cases}.
\]
However, $\partial_2 f_1(0,0)=0$ and $\partial_2 g_1(0,0)=1$.
This means, there is a positive arrow $x_2\to x_1$ in $Gg(0,0)$ while there is no arrow in $Gf(0,0)$.
\end{ex}

%%%%%%%%%%%%%%%%%
\subsection{A mapping from multilevel to Boolean networks}\label{conversion}
Fix the set of states $M$ and a natural number $l$.
We consider mappings from the set $\Net(M)$ of asymptotic evolution functions on $M$
to the set $\Net(\B^l)$ of $l$-dimensional Boolean evolution functions.
Mappings between grid graphs are obtained from them by the correspondence given in Proposition \ref{graph-function-correspondence}.
Following \cite{Didier2011}, we introduce two preferable properties of such mappings.
\begin{dfn}\label{dfn:mapping-prop}
A mapping $\Psi: \Net(M) \to \Net(\B^l)$ is said to be
\begin{itemize}
\item \emph{neighbour-preserving} if there exists a map $b: M \to \B^l$ and $\psi: \B^l\to M$
such that $\psi\circ b$ is the identity on $M$ and 
$b$ and $\psi$ induce graph homomorphisms $\tilde{b}: \Gamma(f)\to \Gamma(\Psi(f))$ and $\tilde{\psi}: \Gamma(\Psi(f)) \to \Gamma(f)$
for any $f\in \Net(M)$.
\item \emph{globally regulation-preserving} if $G\Psi(f)(\B^l)\neq G\Psi(f')(\B^l)$ for any 
$f, f'\in \Net(M)$ with $Gf(M)\neq Gf'(M)$.
\item \emph{locally regulation-preserving} if there exists a map $b: M \to \B^l$
such that $G\Psi(f)(b(y))\neq G\Psi(f')(b(y))$ for any $y\in M$ and
any $f, f'\in \Net(M)$ with $Gf(y)\neq Gf'(y)$.
\end{itemize}
\end{dfn}
These properties are practically useful.
For a neighbour-preserving mapping, the two maps $b$ and $\psi$ give correspondence between the multilevel states and the Boolean states
in such a way that the state transition graph of any multilevel model is embedded in that of a Boolean model.
With a regulation-preserving mapping, one can recover the interaction graph of a multilevel network from the corresponding Boolean one.

A naive idea to convert an evolution function $f: M\to M$ to a Boolean one
 is to use an embedding (one-to-one map) $b: M \to \B^{l}$ of the set of multilevel states
to a higher dimensional set of Boolean states.
Then, the {\em conjugate} of $f$ with respect to $b$ is defined as
\[
f_{b}(x):=b\circ f\circ b^{-1}(x),
\]
which is defined only on $\mathrm{Im}(b)\subset \B^{l}$, the image of $b$.
The domain $\mathrm{Im}(b)$ is called the \emph{admissible region} for $f_b$.
{The state transition graph $\Gamma(f_b)$ in this case is defined to be the full sub-graph 
on $\mathrm{Im}(b)$ of the one defined by Eq. \eqref{eq:stg}.}
Van Ham \cite{VanHam} proposed one particular embedding $b_0: M\to \B^{m_1+m_2+\cdots+m_n}$
which is defined as the direct product of
\begin{equation}\label{VanHam}
 (b_0)_i: \{0,1,\ldots,m_i\}\to \B^{m_i}, \qquad (b_0)_i(k)=(\underbrace{1,1,\ldots,1}_k, 0,\ldots,0) \text{ for all } i\in I.
\end{equation}

Didier et al. \cite{Didier2011} showed that Van Ham's embedding
 is essentially the only one satisfying nice properties which they call
 {\em neighbour preservation} and {\em regulation preservation} (see Remark \ref{rem:didier} below).
On the other hand, an apparent inconvenience of this method
 is that the
resulting evolution function is defined only on the restricted domain $\mathrm{Im}(b)$, 
the set of admissible states \cite{Didier2011}.
In contrast, we will give a construction which produces a Boolean network defined on
the whole state space $\B^{m_1+\cdots+m_n}$.
The idea is to use a surjective map $\psi: \B^{m_1+\cdots+m_n}\to M$ rather than an embedding in the opposite direction.

\begin{rem}\label{rem:didier}
Properties similar to the first two in Definition \ref{dfn:mapping-prop} were introduced by Didier \emph{et al.}~\cite{Didier2011}
but only for embeddings $b: M\to \B^{l}$ (and mappings obtained by conjugation with embeddings Eq. \eqref{VanHam}).
Recall that an embedding $b: M\to \B^{l}$ is said to be neighbour-preserving if it satisfies
$d(b(y),b(y'))=1$ for any $y,y'\in M$ with $d(y,y')=1$.
Also an embedding is said to be regulation-preserving if
$Gf_b(\B^{l})\neq Gf'_b(\B^{l})$ when $Gf(M) \neq Gf'(M)$ for any $f,f': M\to M$;
that is, the global interaction graphs of the Boolean networks obtained by conjugation differ 
when so do those of the multilevel networks.
Our definitions are modified versions of theirs which apply to any mapping.
\end{rem}

Here, we define a mapping from $\Net(M)$ to $\Net(\B^{l})$ with $l=(m_1+\cdots+m_n)$,
and a mapping from grid graphs over $M$ to grid graphs over $\B^{l}$,
which possesses all three above properties.

Define a map $\psi: \B^{l} \to M$ by
\[
\psi(x_{1,1},x_{1,2},\ldots,x_{1,m_1},x_{2,1},\ldots, x_{n,m_n}) =(|y_1|,\ldots,|y_{n}|),
\]
where $y_i=(x_{i,1},\ldots,x_{i,m_i})\in \B^{m_i}$ and $|y_i|=\sum_{k=1}^{m_i} x_{i,k}$.
We denote the index set of $\B^{l}=\B^{m_1+\cdots+m_n}$ by $I_\B=\{(i,j_i)\mid 1\le i\le n, 1\le j_i \le m_i\}$.

\begin{dfn}\label{dfn:binarisation}
For an asymptotic multilevel evolution function $f\in \Net(M)$, its \emph{binarisation} $\B(f)\in \Net(\B^l)$ is defined by
\[
\B(f)_{i,j}(x):=\begin{cases} 0 & (f_i(\psi(x))<\psi(x)_i) \\
x_{i,j} & (f_i(\psi(x))=\psi(x)_i) \\
1 & (f_i(\psi(x))>\psi(x)_i). \end{cases}
\]
Conversely, we have
\begin{equation}\label{eq:bin2m}
f_i(\psi(x)) = \sum_{j=1}^{m_i} \B(f)_{i,j}(x).
\end{equation}

The binarisation of a grid graph $\Gamma$ on $M$, denoted by $\B(\Gamma)$,
is defined to be the grid graph on $\B^l$ such that
there exists a directed edge $x\to x'$ in $\B(\Gamma)$
if and only if $d(x,x')=1$ and there exists a directed edge $\psi(x)\to \psi(x')$ in $\Gamma$.
\end{dfn}
It is trivial to see $\B(f^\Gamma)=f^{\B(\Gamma)}$ and $\B(\Gamma(f))=\Gamma(\B(f))$.
We now identify the image of binarisation $\Net(M)\to \Net(\B^l)$.
The symmetric group $S_{m_i}$ acts on $\B^{m_i}$ by permuting the coordinates.
We consider the coordinate-wise action of $\Sl=S_{m_1}\times S_{m_2} \times \cdots \times S_{m_n}$ on $\B^{l}=\B^{m_1}\times \B^{m_2}\times\cdots \B^{m_n}$.
Since the map $\psi$ is invariant under this action,
the binarisation $\B(f)$ has symmetry with respect to this action.
\begin{dfn}
A Boolean network $f': \B^{l}\to \B^{l}$ is said to be \emph{$\Sl$-symmetric} if
\[
f'_{i,j}(x_{1,1},,\ldots, x_{n,m})=
f'_{i,\sigma_i(j)}(x_{1,\sigma_1(1)},x_{1,\sigma_1(2)},\ldots,x_{1,\sigma_1(m_1)},x_{2,\sigma_2(1)},\ldots, x_{n,\sigma_{n}(m_n)})
\]
 for any $\sigma=(\sigma_1,\ldots,\sigma_n)\in \Sl$.
Similarly, a Boolean grid graph $\Gamma'$  is said to be $\Sl$-symmetric
when an edge $x\to x'$ exists if and only if so does $\sigma(x)\to \sigma(x')$.
\end{dfn}
For an $\Sl$-symmetric Boolean evolution function $f'$ 
we obtain a well-defined evolution function $\psi \circ f' \circ \psi^{-1}$. 
Similarly, 
for an $\Sl$-symmetric $\Gamma'$, we obtain a grid graph over $M$ as the image under $\psi$.
\begin{thm}\label{thm:main}
The binarisation induces a bijective mapping
between the set $\Net(M)$ of asymptotic evolution functions on $M$ and
the set $\Net(\B^{m_1+\cdots+m_n})$ of Boolean evolution functions on $\B^{m_1+\cdots+m_n}$ which are $\Sl$-symmetric.
\end{thm}

It immediately follows that Van Ham's embedding $b_0: M\to \B^{l}$ and our own map $\psi: \B^{l}\to M$
induce graph homomorphisms $\widetilde{b_0}: \Gamma(f) \to \Gamma(\B(f))$ and 
$\tilde{\psi}: \Gamma(\B(f)) \to \Gamma(f)$ for any $f\in \Net(M)$ such that $\tilde{\psi}\circ \widetilde{b_0}$ is the identity.
Thus, the binarisation is neighbour-preserving. We will see that it is also locally, and hence globally as well, regulation-preserving.
We show that the dynamics of the system, namely, attractors in the state transition graph are preserved under binarisation.

In what follows, we often make use of the following two obvious facts:
\begin{lem}\label{lem:Gamma}{\ }\\
\begin{itemize}
\item When there exists an edge $x\to x'$ in $\B(\Gamma)$,
 there exists an edge $\psi(x)\to \psi(x')$ in $\Gamma$.
\item When there exists an edge $y\to y'$ in $\Gamma$,
for any $x\in \psi^{-1}(y)$
 there exists an edge $x\to x'$ in $\B(\Gamma)$ for some $x'\in \psi^{-1}(y')$. 
\end{itemize}
\end{lem}

\begin{prop}\label{prop:scc}
The strongly connected components of $\B(\Gamma)$
map surjectively onto those of $\Gamma$ via $\tilde\psi$.
Moreover, attractors of $\B(\Gamma)$ map surjectively onto those of $\Gamma$ via $\tilde\psi$.
\end{prop}
\begin{proof}
Assume that $x,x'\in \B(\Gamma)$ are in the same strongly connected component.
This means, there is a cycle containing $x,x'$ and it maps to a cycle 
containing $\psi(x),\psi(x')\in \Gamma$. Therefore, $\psi(x),\psi(x')$ are in the same strongly connected component.
Conversely, assume that there exists a cycle containing $y,y'\in \Gamma$.
For any vertex $x\in \psi^{-1}(y)$, there exists a vertex $x'\in \psi^{-1}(y')$
and a cycle containing both $x$ and $x'$. 
To sum up, the image of a strongly connected component of $\B(\Gamma)$
is a strongly connected component of $\Gamma$,
and for any strongly connected component of $\Gamma$ there exists a 
strongly connected component of $\B(\Gamma)$ which maps to it.
Since $\tilde\psi$ is a surjective graph homomorphism,
attractors of $\B(\Gamma)$ map surjectively onto those of $\Gamma$ via $\tilde\psi$.
\end{proof}

\begin{cor}\label{cor:attractors}{\ }\\
\begin{itemize}
\item A stable state exists in $\B(\Gamma)$ if and only if it does in $\Gamma$.
\item A cyclic attractor exists in $\B(\Gamma)$ if and only if it does in $\Gamma$.
\end{itemize}
\end{cor}
\begin{proof}
The first statement is trivial.
For any cyclic attractor in $\Gamma$, there exists an attractor in $\B(\Gamma)$ which maps to it by the previous proposition.
Since it contains more than one element, it is a cyclic attractor in $\B(\Gamma)$.
Conversely, assume that there is a cyclic attractor in $\B(\Gamma)$. 
It contains at least two elements $x,x'$ with $d(x,x')=1$.
Their images $\psi(x),\psi(x')$ should be different since 
any two distinct elements in a single fibre (the inverse image of a point) $\psi^{-1}(\psi(x))$ have at least distance two. 
Thus, the image of the cyclic attractor contains at least two distinct elements $\psi(x),\psi(x')$ 
and is a cyclic attractor in $\Gamma$.
\end{proof}

\begin{thm}\label{thm:regulatry}
The map $I_\B\to I$ defined by $(i,j_i)\mapsto i$ induces a surjective graph homomorphism 
on $G\B(f)(x) \to Gf(y)$ for $y=\psi(x)$ and any $x\in \B^{l}$. More precisely, the following two statements hold.
\begin{enumerate}
\item At any $y\in M$, if a positive (resp. negative) edge $i\to i'$ exists in $Gf(y)$,
so does a positive (resp. negative) edge $(i,j)\to (i',j')$  in $G\B(f)(x)$ for some $j$ and $j'$
 at any $x\in \psi^{-1}(y)$.
\item At any $x\in \B^{l}$, if 
a positive (resp. negative) edge $(i,j)\to (i',j')$ exists in $G\B(f)(x)$,
so does a positive (resp. negative) edge $i \to i'$ in $Gf(\psi(x))$.
\end{enumerate}
\end{thm}
\begin{proof}

\begin{enumerate}
We only show the statements for the case of a positive edge, as the case of a negative edge follows by a similar argument.

For the first statement, assume that there exists a positive edge $i\to i'$ in $Gf(y)$.
We have two cases $\partial^+_i f_{i'}(y)>0$ and $\partial_i^- f_{i'}(y)>0$.
When $\partial^-_i f_{i'}(y)=f_{i'}(y)-f_{i'}(y-e_i)>0$,
for any $x\in \psi^{-1}(y)$ there exists $j$ such that $x_{i,j}=1$ since $y_i>0$.
By Eq. \eqref{eq:bin2m} there must exist some $j'$ such that $\B(f)_{i',j'}(x)-\B(f)_{i',j'}(x-e_{i,j})=1$.
This means there exists a positive edge $(i,j)\to (i',j')$ in $G\B(f)(x)$.
A similar argument applies when $\partial_i^+ f_{i'}(y)>0$.

For the second statement, assume that there exists a positive edge $(i,j)\to (i',j')$ in $G\B(f)(x)$.
When $x_{i,j}=0$, this means 
$\B(f)_{i',j'}(x)=0$ and $\B(f)_{i',j'}(x+e_{i,j})=1$.
Since $\psi(x)+e_{i}=\psi(x+e_{i,j})$, we have
 $f_{i'}(\psi(x)+e_i)-f_{i'}(\psi(x))=\sum_{k=1}^{m_{i'}} \left(\B(f)_{i',k}(x+e_{i,j}) -  \B(f)_{i',k}(x) \right)$.
 Since $\B(f)_{i',k}(x+e_{i,j}) -  \B(f)_{i',k}(x) \ge 0$ for all $k$ and $\B(f)_{i',j'}(x+e_{i,j}) -  \B(f)_{i',j'}(x)=1$,
 we have $f_{i'}(\psi(x)+e_i)-f_{i'}(\psi(x))>0$.
This in turn means that there exists a positive edge $i\to i'$ in $Gf(\psi(x))$.
A similar argument applies to the case when $x_{i,j}=1$.
\end{enumerate}
\end{proof}
Intuitively speaking, (1) says all the regulation in the original multilevel network is captured
 in the converted Boolean network, while
(2) says all the regulation in the converted network comes from the original multilevel network.

\begin{cor}\label{cor:circuit}
An asymptotic evolution function $f$ over $M$ has a positive (negative)
type-1 functional circuit if so does its binarisation $\B(f)$.
\end{cor}
Note that two negative arrows $(i,j)\to (i,j')$ and $(i,j')\to (i,j)$ in $G\B(f)(x)$,
both of which correspond to a negative self-loop $i\to i$ in $G\B(f)(x)$,
can be composed to produce a positive circuit.
This positive functional type-1 circuit corresponds
the one which is the composition of a single negative self-circuit with itself in $Gf(\psi(x))$.

%Furthermore, $f$ has a positive type-1 functional circuit if 
%its binarisation $\B(f)$ has one which is not 
%the composition of two negative arrows of the form $(i,j)\to (i,j')\to (i,j)\in G\B(f)(x)$.

\begin{ex}
We give two characteristic examples of the binarisation.

Consider the evolution function $f(y)=2-y$ over $M=\{0,1,2\}$.
Its binarisation is 
\[
\B(f)(x)=\begin{cases} (1,1) & (x=(0,0)) \\
x & (x=(1,0), (0,1)) \\ (0,0) & (x=(1,1)). \end{cases}
\]
The corresponding state transition graphs are
\[
\xymatrix{
{0} \ar[r] & 1 & 2\ar[l]
}, 
\qquad
\xymatrix{
& 10 & \\
{00} \ar[ru] \ar[rd] & & 11 \ar[lu] \ar[ld] \\
& 01 & 
}
\]
The interaction graphs at $y=0$ and $x=(0,0)$ respectively look:
\[
Gf(0)=
\xymatrix{
y \ar@{->}@(ul,ur)^{-}
},
\qquad
G\B(f)(0,0)=\xymatrix{
x_{11} \ar@/^1pc/[r]^{-} & x_{12} \ar@/^1pc/[l]^{-}
}
\]
Notice that the self-loop on $y_1$ corresponds to each of the edges $x_{11}\to x_{12}$ and $x_{11}\leftarrow x_{12}$.
In particular, the converse to Corollary \ref{cor:circuit} does not hold.

Consider another evolution function  over $M=\{0,1,2\}$ defined by $f(y)=\begin{cases} y & (y=0,1) \\ 0 & (y=2). \end{cases}$
Its binarisation is $\B(f)(x)=\begin{cases} x & (x\neq(1,1)) \\ (0,0) & (x=(1,1)). \end{cases}$
The corresponding state transition graphs are
\[
\xymatrix{
0 \ar@{--}[r] & 1 & 2\ar[l]
}, 
\qquad
\xymatrix{
& 10 & \\
{00} \ar@{--}[ru] \ar@{--}[rd] & & {11} \ar[lu] \ar[ld] \\
& 01 & 
}
\]
The interaction graphs at $y=1$ and $x=(0,1)$  respectively look:
\[
Gf(1)=
\xymatrix{
y \ar@{->}@(ul,ur)^{-} \ar@{->}@(dl,dr)_{+}
},
\qquad
G\B(f)(0,1)=
\xymatrix{
x_{11} \ar@/^1pc/[r]^{-}  & x_{12}  \ar@{->}@(dr,ur)_{+}
}
\]
\end{ex}

\subsection{Another extension method}
Recently, Tonello \cite{Tonello2017} has independently
constructed a mapping which also extends Van Ham's while preserving the dynamics 
and the local regulations in a more stringent sense than ours.
Her method was used to produce a counter-example to Conjecture \ref{negative-circuit-conjecture} as well.
Her mapping can be described in our context as follows:
\[
 f \mapsto b_0 \circ f \circ \psi,
\]
where $f$ is a stepwise function.
%Her method has some advantages over ours in that it
{Compared to ours, her method} yields {fewer} %less 
arrows in the state transition graph.
Her strategy was to stipulate
the converted function to take values in the admissible region $\mathrm{Im}(b_0)$,
 whereas ours was to equip the converted function with the symmetry described in Theorem \ref{thm:main}.

%%%%%%%%%%%%%%%%%%%
\section{Results}
%%%%%%%%%%%%%%%%%%%%
\subsection{Lambda phage}
As an illustration, we first apply our method to the 2-variable lambda phage model proposed by Thieffry and Thomas \cite{Thieffry1995}.

The lambda phage is a bacterial virus that infects \emph{E. coli}. It is a temperate phage, \emph{i.e.} it can either multiply and eventually kill the host cell (\emph{lytic} phase), or integrate its DNA into the bacterial chromosome (\emph{lysogenic} phase), conferring the cell immunity against super-infection by other lambda phage. The switch between lysis and lysogeny, as modelled by Thieffry and Thomas, is essentially controlled by a positive feedback circuit between genes $cI$ and $cro$. The two genes inhibit each other, such that $cI$ dominates the lysogenic phase, whereas $cro$ is active during the lytic phase. The gene cro further inhibits its own activity. The system is modelled by a discrete system with the state space $M=\{(cI,cro)\in \{0,1\}\times \{0,1,2\} \}$. The dynamics displays a stable state with high $cI$ and low $cro$ activity, and a two-state cyclic attractor with low $cI$, and $cro$ oscillating around its activity threshold (Fig.~\ref{fig:Phage2STG}, left), which the authors describe as homeostasis.

\begin{figure}[!ht]
   \includegraphics[width=0.8\linewidth,keepaspectratio=true]{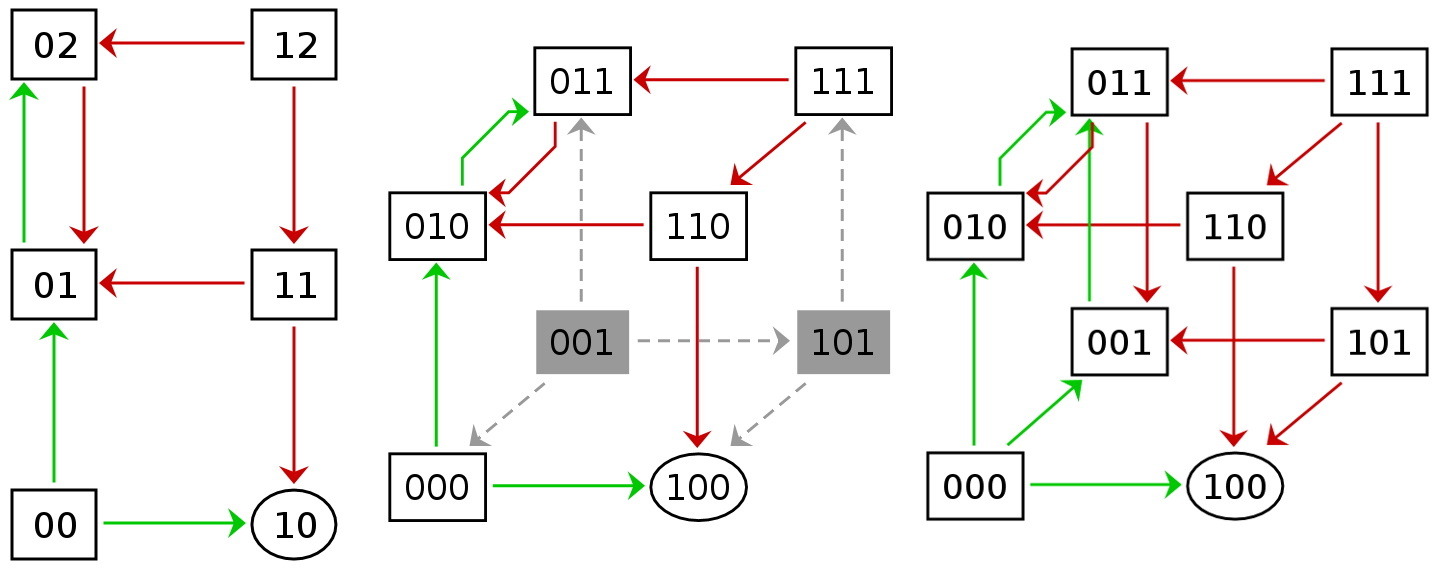}
 \caption{\csentence{State transition graphs for the lambda phage model.}
  Left, the original, multilevel model \cite{Thieffry1995}; centre, the Boolean version obtained using Van Ham's method; right, the Boolean version obtained using our method.}    
  \label{fig:Phage2STG}
\end{figure}

Table \ref{tab:lambdaPhage} shows how the same dynamics can be encoded using a stepwise or asymptotic evolution function. Notice that the stepwise function creates a positive feedback on $cro$ ($\hat{f}_{cro}(0,1)-\hat{f}_{cro}(0,0)>0$, and $\hat{f}_{cro}(1,2)-\hat{f}_{cro}(1,1)>0$) that is not visible in the asymptotic function. This difference in the global regulatory graphs is shown in Fig.~\ref{fig:Phage2global}.

\begin{table}[h!]
\centering
\caption{Truth table of the lambda phage model} \label{tab:lambdaPhage}
\begin{tabular}{cc|cc|cc}
$cI$ & $cro$ & $\hat{f}_{cI}$ & $\hat{f}_{cro}$ & $\bar{f}_{cI}$ & $\bar{f}_{cro}$ \\
\hline
0 & 0 & 1 & 1 & 1 & 2 \\
0 & 1 & 0 & 2 & 0 & 2 \\
0 & 2 & 0 & 1 & 0 & 0 \\
1 & 0 & 1 & 0 & 1 & 0 \\
1 & 1 & 0 & 0 & 0 & 0 \\
1 & 2 & 0 & 1 & 0 & 0 \\
\end{tabular}
\begin{flushleft} Table gives the target level of $cI$ and $cro$ in each state, encoded using either the stepwise or asymptotic evolution function. 
\end{flushleft}
\label{table1}
\end{table}

Boolean systems are generated by Van Ham's method and ours with the state space $\{(cI,cro1,cro2)\in \B^3\}$. 
However, since Van Ham's method yields a dynamics with as many states as the original, multilevel dynamics (Fig.~\ref{fig:Phage2STG}, centre), the system thus obtained includes a ''non-admissible'' region (grey area in the Figure) whose states do not have any counterpart in the multilevel model.
To make comparison, we extended the domain of the Boolean model obtained by Van Ham's method
(Fig.~\ref{fig:Phage2STG}, centre) by completing the grey dashed arrows in such a way that 
\begin{itemize}
\item it does not create any extra arrows in the global interaction graph that is not already visible elsewhere within the admissible region
\item and there is no outgoing arrow in the state transition graph from an admissible state to a non-admissible state.
\end{itemize}
This extension is based on our understanding of Van Ham's original publication~\cite{VanHam}.
The corresponding global regulatory graph is shown in Fig.~\ref{fig:Phage2global} (centre).

\begin{figure}[!ht]
   \includegraphics[width=0.6\linewidth,keepaspectratio=true]{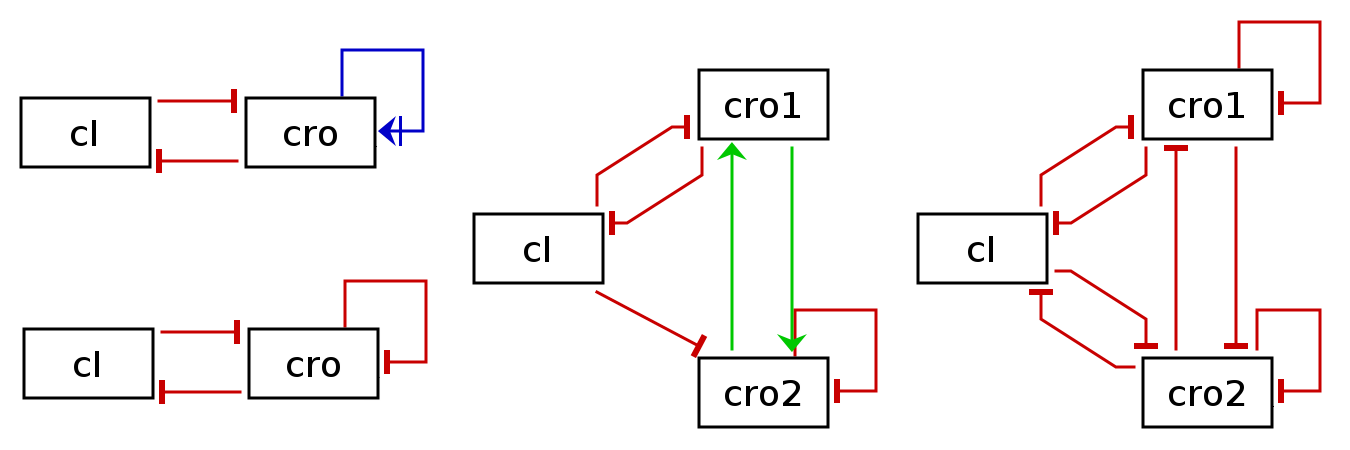}
 \caption{\csentence{Global interaction graphs for the lambda phage model.}
 Left, the original, multilevel model \cite{Thieffry1995}: top, stepwise evolution function, bottom, asymptotic evolution function; centre, the Boolean version obtained using Van Ham's method; right, the Boolean version obtained using our method.
 }  \label{fig:Phage2global}
\end{figure}

In contrast, the dynamics produced by our method is more complex, and it occupies the whole state space:
every state has a counterpart in the original multilevel model. However, the two-state attractor of the original model is now represented by a three-state attractor, where state $011$ corresponds to state $02$ and both states $001$ and $010$ correspond to state $01$. Another difference is that, in the model obtained with Van Ham's method, variables $cro1$ and $cro2$ are ordered and represent levels 1 and 2, respectively, of the original multilevel $cro$. In the Boolean model generated with our method, $cro1$ and $cro2$ are interchangeable and equivalent: whether one represents level 1 or 2 depends on the other, and thus, on the context of each particular state.

Finally, an important difference appears at the local level (Fig.~ \ref{fig:Phage2local}). Using Van Ham's method, local graphs may include edges that have no visible counterpart in the local graph of the corresponding multilevel state. For example, while in $02$ the only visible regulation is the negative loop on $cro$, Van Ham's method adds an edge from $cro1$ to $cI$ in $011$, whereas the corresponding local graph obtained using our method includes only regulations between $cro1$ and $cro2$. Similarly, in $10$, the only visible edges occur between $cI$ and $cro$, and the same is true in $100$ using our method; however, using Van Ham's method an additional edge between $cro1$ and $cro2$ becomes visible.
{Our method generates an ``extra'' positive circuit between $cro1$ and $cro2$,
which in multilevel model corresponds to the composition of the negative self circuit on $cro$ with itself.
Such positive circuits can only appear between variables that represent the same multilevel variable
(see Corollary \ref{cor:circuit}).}

\begin{figure}[!ht]
   \includegraphics[width=0.6\linewidth,keepaspectratio=true]{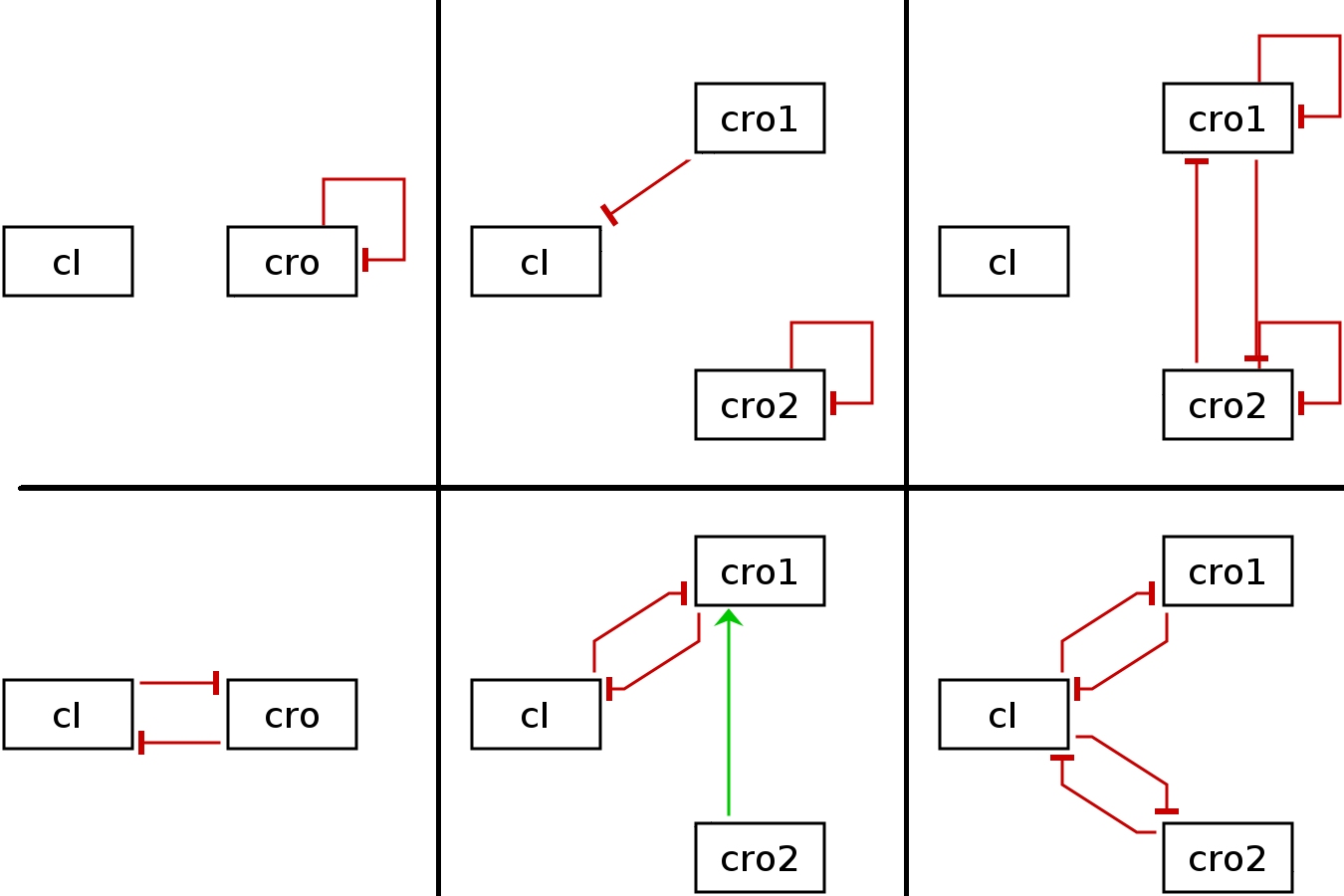}
 \caption{\csentence{Local interaction graphs for the lambda phage model.}
 Left, the original, multilevel model \cite{Thieffry1995}; centre, the Boolean version obtained using Van Ham's method; right, the Boolean version obtained using our method. Top, local graphs for low $cI$ and high $cro$ (states $02$ or $011$); bottom, local graphs for states with high $cI$ and low $cro$ (states $10$ and $100$). Local graphs in $02$ and $10$ are identical under the stepwise or asymptotic evolution function.}    \label{fig:Phage2local}
\end{figure}

%%%%%%%%%%%%%%
\subsection{Boolean counter-example to Conjecture \ref{negative-circuit-conjecture}}\label{ex:counter-example}
One of the main goals in genetic regulatory network analysis is
to find relations between circuits in the interaction graphs and 
attractors in the state transition graphs.
Here we list a few known results in this direction:

	\begin{enumerate}
		\item If $f$ has no type-1 functional circuit, $\Gamma(f)$ has a unique stable state $x$ (\cite{Shih2005}).
%		Moreover, there is a path from any $y\in M$ to $x$ of length equal to $d(y,x)$, where $d(y,x)$ is the Hamming distance		
		\item If $f$ has no type-1 functional positive circuit, $\Gamma(f)$ has a unique attractor (\cite{Richard2007,Remy2008}).
		\item If there exists a cyclic attractor in $\Gamma(f)$, $Gf(M)$ must contain a negative circuit
		(\cite{Richard2010}).
	\end{enumerate}
%\red{Remember that $\Gamma(f)$ is the asynchronous state transition graph. Indeed, it is trivial to show that synchronous simulation of an isolated positive or negative circuit may yield multiple cyclic attractors (see \cite{Remy2003} for a detailed description of the dynamical properties of isolated circuits).}  <= what's the point?
Notice that the first two statements connect the asymptotic behaviour of the dynamics with
the {\em local} interaction graph at some $x\in M$, while the last one does so with the \emph{global} interaction graph $Gf(M)$.
This naturally gives rise to the following conjectures:
\begin{conj}[{\cite[Question 1]{Comet2013},\cite[Question 1,2]{Richard2010}}]\label{negative-circuit-conjecture} {\ }\\
\begin{enumerate}
\item If $f$ has no type-1 functional negative circuit, then $\Gamma(f)$ has no cyclic attractors.
\item If $f$ has no type-1 functional negative circuit, then $\Gamma(f)$ has a stable state.
\end{enumerate}
\end{conj}
Note that the second statement is weaker in the sense that it follows from the first one.

Richard, together with Comet \cite[Example 6]{Richard2010}, gave a counter-example to %the second, and thus both of,
the conjectures when $M=\{0,1,2,3\}^2$. 
The grid graph $\Gamma$ over $M=\{0,1,2,3\}^2$ given in \cite[Example 6]{Richard2010}
has no stable state and there exists no type-1 negative functional circuit in the interaction graph of $\bar{f}^\Gamma$ 
(see Fig. \ref{fig:richard}).
The existence of Boolean counter-examples was left as an open problem in \cite{Richard2010}.

\begin{figure}[!ht]
   \includegraphics[width=0.4\linewidth,keepaspectratio=true]{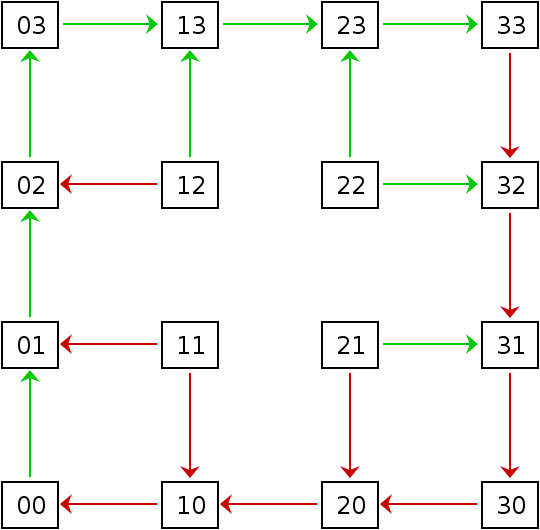}
 \caption{\csentence{A multilevel counter-example given in \cite{Richard2010}}}\label{fig:richard}
\end{figure}

By Corollary \ref{cor:circuit}, the Boolean grid graph $\B(\Gamma)$ yielded by our method (Fig. \ref{fig:ours})
has no type-1 negative functional circuit in the interaction graph of the binarisation $\B(\bar{f}^\Gamma)$.
Furthermore, by Corollary \ref{cor:attractors} there is no stable state in the state transition graph of $\B(\bar{f}^\Gamma)$.
Thus, we obtain a new Boolean counter-example to Conjecture \ref{negative-circuit-conjecture}.

\begin{figure}[!ht]
   \includegraphics[width=1.0\linewidth,keepaspectratio=true]{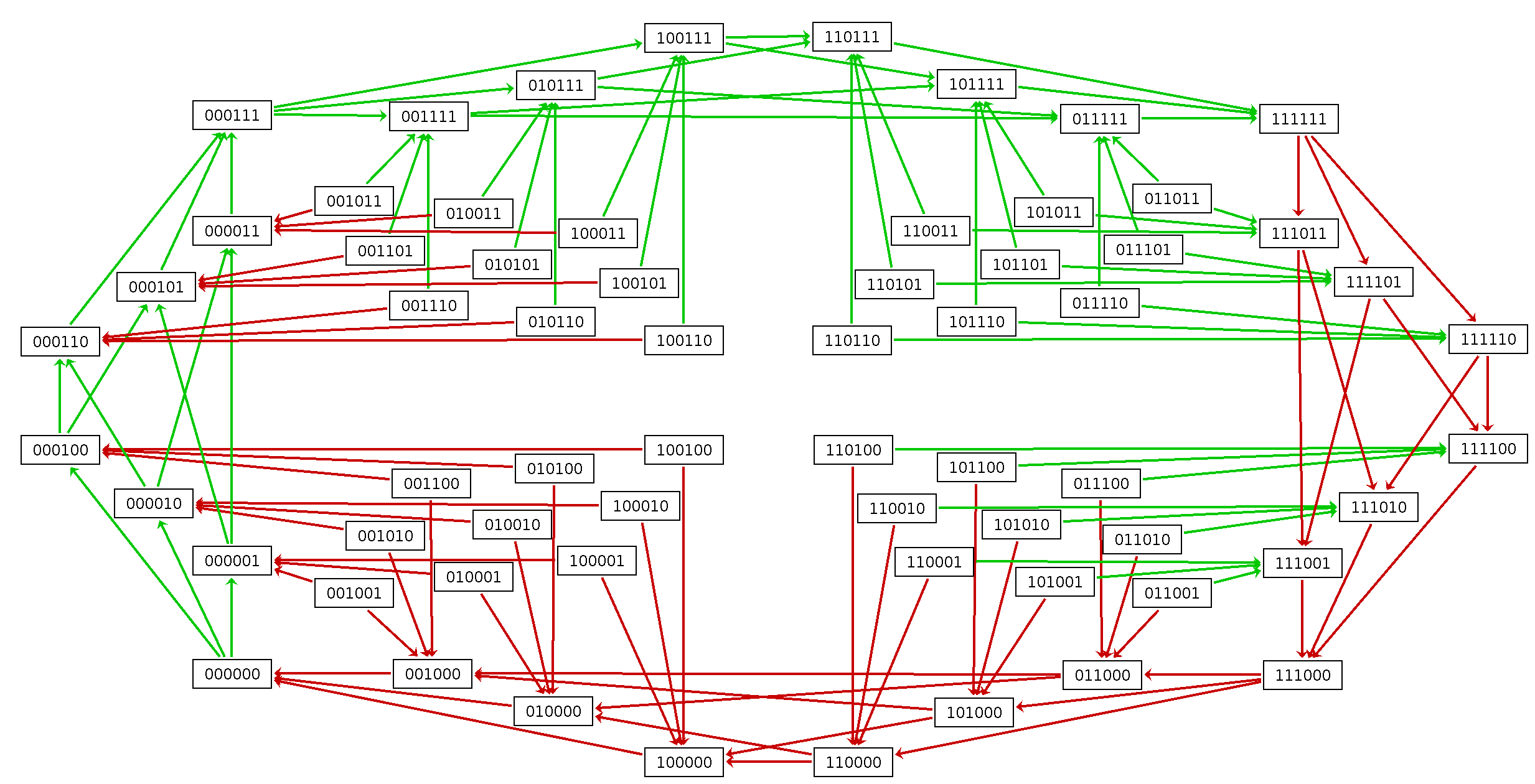}
 \caption{\csentence{A Boolean counter-example to Conjecture \ref{negative-circuit-conjecture}}
 The binarisation of the multilevel counter-examples gives a Boolean counter-example}\label{fig:ours}
\end{figure}

Finally, we note that Ruet recently gave a systematic way to produce counter-examples to the conjectures for the Boolean case~\cite{Ruet2017}.
Our method is very different from Ruet's and while his example possesses a special property that
it has an {\em attractive cycle}, our counter-example has a smaller dimension of $6$.
For comparison, we include an example based on Ruet's theory \cite{Ruet2017} Fig.~\ref{SI_Fig}.

\begin{figure}[!ht]
   \includegraphics[width=1.0\linewidth,keepaspectratio=true]{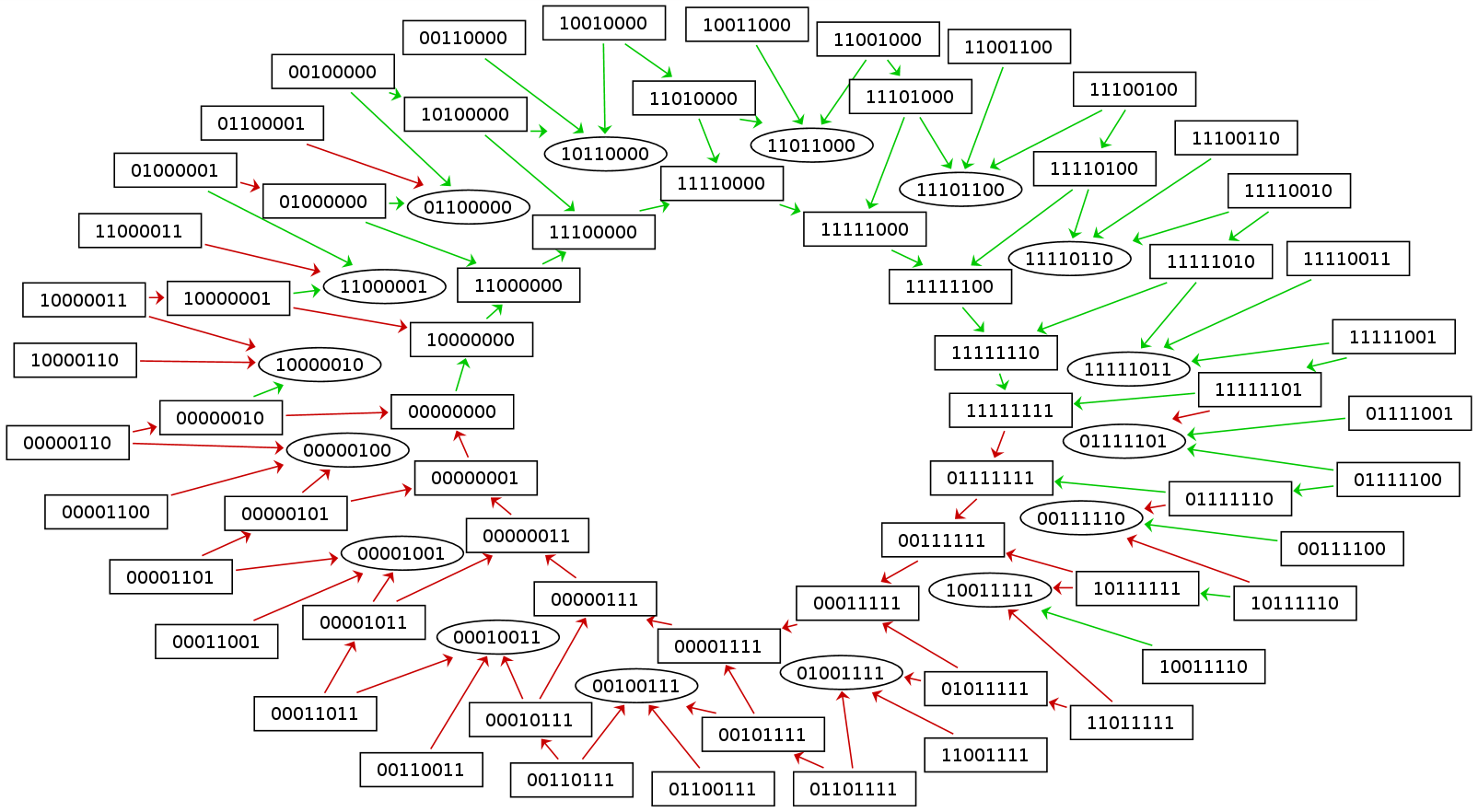}
 \caption{\csentence{Antipodal attractive cycle with no negative circuit.}
  States not shown are all stable states. Adapted from \cite{Ruet2017}.}\label{SI_Fig}
\end{figure}

%%%%%%%%%%%%%%
\subsection{Implementation}\label{sec:implementation}
%%%%%%%%%%%%%%
We implemented our method in the form of a Perl script.
For comparison, Tonello's method \cite{Tonello2017} is also implemented in the same script.
It is available at \cite{Kaji2017} under the MIT license.
The input is a multilevel evolution function described in the Truth table format \cite{Naldi2009}
and the output is a Boolean evolution function described in the same format.
The maximum levels $m_i$ for each gene is automatically detected.
When the input function is defined on $\{0,1,\ldots,m_1\}\times \{0,1,\ldots,m_2\} \times \cdots \times \{0,1,\ldots,m_n\}$,
the converted Boolean network is defined on $\B^{m_1+m_2+\cdots+m_n}$.

%%%

%%%%%%%%%%%%%%
\section{Discussion and conclusions}
%%%%%%%%%%%%%%
In the discrete formulation of gene regulatory networks,
a system is commonly modelled by a function.
When some genes take more than two levels,
there are multiple choices for functions having the same (asynchronous) state transition graph.
We single out a unique choice, which we call the asymptotic evolution function (Proposition \ref{graph-function-correspondence}).
Then, we introduced a mapping which converts an asymptotic evolution function to a Boolean evolution function (Definition \ref{dfn:binarisation}). This mapping preserves dynamical and regulatory properties
(Proposition \ref{prop:scc}, Theorem \ref{thm:regulatry}), 
thus allowing us to analyse multilevel networks by methods developed for Boolean networks.

Mappings from multilevel to Boolean networks have been used in the study of gene regulatory networks.
In particular, Van Ham's mapping has been shown by Didier \emph{et al.} to be essentially the only method to provide a one-to-one, neighbour-preserving and regulation-preserving Boolean representation of multilevel models \cite{Didier2011}. However, although the authors did suggest that the mapping could be useful to study the role of regulatory circuits, the question of how interaction functionality contexts are preserved had not been studied so far.

One such instance is Thomas's conjecture, which states that the existence of a cyclic attractor {in the asynchronous state transition graph} requires that of local negative circuit.
The conjecture has recently been given a Boolean counter-example by P.~Ruet \cite{Ruet2017}. Until then, although A.~Richard and and J-P.~Comet had produced a multilevel counter-example \cite{Richard2010}, the Boolean case remained open. It is straightforward to apply Van Ham's method to the counter-example in order to obtain a Boolean model, but it is defined only on the admissible region. To extend the model to the whole Boolean state space, while preserving its dynamics and regulatory relation, is highly non-trivial.
The method we propose has been designed specifically to circumvent the problem. The idea was to avoid extra interactions and circuits 
by extending the state transition graph obtained with Van Ham's method 
in such a way that we have ``parallel trajectories'' going through the whole state space.
This was achieved by loosening the one-to-one criterion such that states including intermediate values in the multilevel model would match with several states in the Boolean version, effectively creating equivalent Boolean transitions for each multilevel transition in the original model.

In a sense, our method works opposite to Van Ham's: instead of embedding multilevel states into Boolean states, we define a Boolean model such that each Boolean state can be mapped to a multilevel state.
With our method, interaction functionality is preserved, and thus all local interaction graphs in the Boolean model come from their counterpart in the original, multilevel model. In contrast, Van Ham's method only preserves the global interaction graph.

One limitation of our method is that the synchronous dynamics of the original multilevel model can not be directly retrieved from the Boolean model produced by the conversion. 
Here, by \emph{synchronous dynamics} we mean the state transition graph
having edges $x\to f(x)$ instead of \eqref{eq:stg} (see, for example, \cite{Remy2003}).
Synchronous state transition is often deemed unrealistic
since it assumes all processes are realised simultaneously with the same delay
 (see discussion by Abou-Jaoud\'e \emph{et al.}~\cite{Abou-Jaoude2016}).
Nevertheless, the synchronous mode is still a popular update method in simulation due to its simplicity, 
and occasionally used for multilevel models (see \emph{e.g.} Chifman \emph{et al.} \cite{Chifman2017} for a recent example). 
If our binarisation is used for a synchronous simulation,
any increase in a variable would translated into its increase to the maximum value.
%It is possible to use the conversion $f \mapsto b_0 \circ f \circ \psi$ to circumvent this problem, but then attractors are not preserved any more.
%More importantly, the same problem would apply to stochastic approaches as long as they allow the simultaneous update of more than one variable. In these cases Van Ham's or Tonello's should remain the preferred conversion methods.}
It is worth noting that the 
counter-examples for Conjecture \ref{negative-circuit-conjecture} 
considered in \S \ref{ex:counter-example} (Richard-Comet's one and its binarisation)
have no stable state in the synchronous state transition graph as well
(and have no type-1 negative functional circuit in the interaction graph).

Finally, our results highlight two opposite strategies, stepwise and asymptotic, for writing the evolution function of a multilevel model.
While both suppress inter-genes regulations,  
the stepwise tends to add positive self-regulations, whereas the asymptotic tends to add negative self-regulations.
This work contributes to a better understanding of the different ways to represent a multilevel system,
for different ways can represent the same model \cite{Streck2015}, which causes ambiguities in the notation.

%%%%
%\section{Questions}
%\begin{itemize}
%\item Can we give an easy proof of  Theorem \ref{thm:type-1-functionality}?
%\item Partial derivatives are associated to vertices, whereas gene regulation is to edges.
%This can be a hint to find the {\em right} definition for functionality of circuits.
%\item Instead of looking at functional circuits, can we use ``summarised derivatives'' $\sum_j \partial_j f_i(x)$ ?
%\item Can we say anything about existence of cyclic attractors (and their length) from local data such as degrees of vertices in $\Gamma(f)$ ?
%\item What kind of Jacobians are {\em integrable}? 
%Integrability of a given collection of matrices $\{J_x\}$ means that there exists a network $f$ satisfying 
%$Jf(x)=J_x$ for all $x$ with $x_j<m$.
%(Similarly, what kind of collections of local interaction graphs $\{G_x\}$ are realisable?)
%\end{itemize}
%

\section*{Acknowledgements}
The authors are grateful to Paul Ruet for explaining his result in \cite{Ruet2017},
to Elisa Tonello for fruitful discussion and careful reading of our draft,
and to Yuki Ikawa and Sergey Tishchenko for their help in the early stage of this work.
The second named author is partially supported by JST PRESTO Grant Number JPMJPR16E3, Japan.

\bibliographystyle{bmc-mathphys} % Style BST file (bmc-mathphys, vancouver, spbasic).
\bibliography{grid-dynamics_bmc}      % Bibliography file (usually '*.bib' )

\end{document}